\documentclass[a4paper,12pt]{article}
\usepackage{amsmath}
\usepackage{amssymb}
\usepackage{amsfonts,times}
\usepackage{theorem}
\newtheorem{theorem}{Theorem}

\newtheorem{corollary}[theorem]{Corollary}

\newtheorem{definition}[theorem]{Definition}

\newtheorem{lemma}[theorem]{Lemma}

\newtheorem{proposition}[theorem]{Proposition}
{\theorembodyfont{\upshape}\newtheorem{remark}[theorem]{Remark}}

\newenvironment{proof}[1][Proof]{\noindent\textbf{#1.} }{\ \hfill \rule{0.5em}{0.5em}}

\newcommand{\ui}{\underline{i}}
\newcommand{\cG}{\mathcal{G}}
\begin{document}

\title{Semicircle law for a matrix ensemble with dependent entries}
\author{Winfried Hochst\"{a}ttler, Werner Kirsch \\
Fakult\"{a}t f\"{u}r Mathematik und Informatik\\
FernUniversit\"{a}t in Hagen, Germany \and Simone Warzel \\
Zentrum Mathematik\\
Technische Universit\"{a}t M\"{u}nchen, Germany}
\date{}

\maketitle

\abstract{We study ensembles of random symmetric matrices whose entries
exhibit certain correlations. Examples are distributions of
Curie-Weiss-type. We provide a criterion on the correlations ensuring the
validity of Wigner's semicircle law for the eigenvalue distribution measure.
In case of Curie-Weiss distributions this criterion applies above the
critical temperature (i.~e. $\beta<1$).  We also investigate the largest eigenvalue of
certain ensembles of Curie-Weiss type and find a transition in its behavior at the
critical temperature.\\}

\section{Introduction}

In this article we consider random matrices $X_{N}$ of the form%
\begin{equation}
X_{N}=\left(
\begin{array}{cccc}
X_{N}(1,1) & X_{N}(1,2) & \ldots & X_{N}(1,N) \\
X_{N}(2,1) & X_{N}(2,2) & \ldots & X_{N}(2,N) \\
\vdots & \vdots &  & \vdots \\
X_{N}(N,1) & X_{N}(N,2) & \ldots & X_{N}(N,N)%
\end{array}%
\right)  \label{eq:RM}
\end{equation}

The entries $X_{N}(i,j)$ are real valued random variables
varying with $N$. We will always assume that the matrix $X_{N}$ is
symmetric, such that\goodbreak\noindent $X_{N}(i,j)=X_{N}(j,i)$ for all $i,j$. Furthermore we
suppose that all moments of the $X_{N}(i,j)$ exist and that $\mathbb{E(}%
X_{N}(i,j))=0$ and $\mathbb{E(}X_{N}(i,j)^{2})=1$.

It is convenient to work with the \emph{normalized} version $A_{N}$ of $%
X_{N} $, namely with
\begin{equation}
A_{N}=\frac{1}{\sqrt{N}}X_{N}  \label{RMnorm}
\end{equation}

As $A_{N}$ is symmetric it has exactly $N$ real eigenvalues (counting
multiplicity). We denote them by%
\begin{equation*}
\lambda _{1}\leq \lambda _{2}\leq \ldots \leq \lambda _{N}
\end{equation*}
and define the (empirical) eigenvalue distribution measure by%
\begin{equation*}
\sigma _{N}=\frac{1}{N}\sum_{j=1}^{N}\delta _{\lambda _{j}}
\end{equation*}
and its expected value $\overline{\sigma }_{N},$ the density of states
measure by%
\begin{equation*}
\overline{\sigma }_{N}~=~\mathbb{E\,}\left( \frac{1}{N}\sum_{j=1}^{N}\delta
_{\lambda _{j}}\right) ~~.
\end{equation*}

If the random variables $X_{N}(i,j)$ are independent and identically
distributed (i.i.d.) (except for the symmetry condition $X_N(i,j)=X_N(j,i)$)
then it is well known that the measures $\overline{\sigma }_{N}$ and $\sigma _{N}$ converge weakly to the
semicircle distribution $\sigma _{sc}$ (almost surely in the case of $\sigma _{N}$). The semicircle distribution is
concentrated on the interval $[-2,2]$ and has a density given by $\sigma
_{sc}(x)=\frac{1}{2\pi }\sqrt{4-x^{2}}$ for $x\in \lbrack -2,2]$. This
important result is due to Eugen Wigner \cite{Wigner1} and was proved by
Arnold \cite{Arnold} in greater generality, see also for example \cite%
{Pastur}, \cite{Pastur Sherbina} or \cite{AGZ}.

\bigskip

Recently, there was a number of papers considering random matrices with some
kind of dependence structure among their entries, see for example \cite%
{Chatterjee}, \cite{Goetze Tikhomirov 1}, \cite{Goetze Naumov Tikhomirov} and \cite{Schenker Schulz-Baldes}.
In particular the papers \cite{Bryc et al}, \cite{FriesenLoewe1} and \cite%
{FriesenLoewe2} consider symmetric random matrices whose entries $X_{N}(i,j)$
and $X_{N}(k,\ell )$ are independent if they belong to different diagonals,
i.e. if $\left\vert i-j\right\vert \neq \left\vert k-\ell \right\vert $, but
may be dependent within the diagonals. It was in particular the work \cite%
{FriesenLoewe2} which motivated the current paper. Among other models
Friesen and L\"{o}we \cite{FriesenLoewe2} consider matrices with independent
diagonals and (independent copies of) Curie-Weiss distributed random
variables on the diagonals. (For a definition of the Curie-Weiss model see
below).

The main example for the results in our paper is a symmetric random matrix
whose entries $X_{N}(i,j)$ are Curie-Weiss distributed for all $i,j$ (with $%
i\le j$). The models considered in this paper also include the Curie-Weiss model on diagonals investigated by
Friesen and L\"{o}we. For the reader's convenience we define our Curie-Weiss
ensemble here, but we'll work with abstract assumptions in the following two
chapters.

\bigskip In statistical physics the \emph{Curie-Weiss} model serves as the
easiest nontrivial model of magnetism. There are $M$ sites with random
variables $X_{i}$ attached to the sites $i$ taking values $+1$ ("spin up") or $%
-1 $ ("spin down"). Each spin $X_{i}$ interacts with all the other spins
prefering to be aligned with the average spin $\frac{1}{M}\sum_{j\neq
i}X_{j}.$ More precisely:

\begin{definition}
\bigskip Random variables $\{X_{j}\}_{j=1,\ldots ,M}$ with values in $%
\{-1,+1\}$ are distributed according to a Curie-Weiss law $\mathbb{P}_{\beta
,M}$ with parameters $\beta \geq 0$ (called the inverse temperature) and $%
M\in
\mathbb{N}
$ (called the number of spins) if
\begin{equation}
\mathbb{P}_{\beta ,M}(X_{1}=\xi _{1},X_{2}=\xi _{2},\ldots ,X_{M\,}=\xi
_{M})~=~ Z_{\beta ,M}^{-1}\,\,\frac{1}{2^{M}}\,e^{\frac{\beta }{2M}%
(\sum \xi _{j})^{2}}
\end{equation}%
where $\xi _{i}\in \{-1,+1\}$ and $Z_{\beta ,M}$ is a normalization constant.
\end{definition}

For $\beta <1$ Curie-Weiss distributed random variables are only weakly
correlated, while for $\beta >1$ they are strongly correlated. This is
expressed for example by the fact that a law of large numbers holds for $%
\beta <1,$ but is wrong for $\beta >1$. This sudden change of behavior is
called a "phase transition" in physics. In theoretical physics jargon the quantity $T=\frac{1}{\beta}$ is called the \emph{temperature} and $T=1$ is called the
\emph{critical temperature}. More information about the Curie-Weiss model and its physical meaning can be found in
\cite{Thompson} and \cite{Ellis}.

Our Curie-Weiss matrix model, which we dub the \emph{full Curie-Weiss
ensemble}, is defined through $M=N^{2}$ random variables $\left\{
Y_{N}(i,j)\right\} _{1\leq i,j\leq N}$ which are $\mathbb{P}_{\beta ,M}$%
-distributed. To form a \emph{symmetric }matrix we set $%
X_{N}(i,j)=Y_{N}(i,j) $ for $i\leq j$ and $X_{N}(i,j)=Y_{N}(j,i)$ for $i>j$
and define

\begin{equation*}
A_{N}=\frac{1}{\sqrt{N}}X_{N}
\end{equation*}

By the \emph{diagonal Curie-Weiss} \emph{ensemble} we mean a symmetric
random matrix with the random variables on the $k^{th}$ diagonal $\left\{
i,i+k\right\} $ being $\mathbb{P}_{\beta ,N}$-distributed ($0\leq k\leq N-1$
and $1\leq i\leq N-k$) and with entries on different diagonals being
independent. This model was considered in \cite{FriesenLoewe2}. For $\beta
<1 $ we will prove the semicircle law for these two ensembles.

In the following section we formulate our general abstract assumptions and
state the first theorem of this paper which establishes the semicircle law
for our models. The proof follows in Section \ref{sec:proof}.

In Section \ref{sec:CW} we discuss our main example, the full Curie-Weiss
model, in fact we will study various random matrix ensembles associated to
Curie-Weiss-like models. In this section we also discuss exchangeable random
variables and their connection with the Curie-Weiss model.

In Section \ref{sec:largeig} we investigate the largest eigenvalue (and thus the matrix norm) of Curie-Weiss-type
matrix ensembles both below and above the critical value $\beta=1$.
\medskip

\medskip \textbf{Acknowlegment }It is a pleasure to thank Matthias L\"{o}we,
M\"{u}nster, and Wolfgang Spitzer, Hagen, for valuable discussion. Two of us
(WK and SW) would like to thank the Institute for Advanced Study in
Princeton, USA, where part of this work was done, for support and
hospitality.

\section{The semicircle law\label{sec:result}}

\begin{definition}
\label{def:SchemeMatrix}
Suppose $\left\{ I_{N}\right\} _{N\in \mathbb{N}}$
is a sequence of finite index sets $I_{N}$. A family $\left\{ X_{N}(\rho
)\right\} _{\rho \in I_{N},N}$ of random variables indexed by $N\in\mathbb{N}$ and (for given $N$)
by the set $I_{N}$ is called an $\left\{ I_{N}\right\} $\emph{%
-scheme of random variables}. If the sequence $\left\{ I_{N}\right\} $ is
clear from the context we simply speak of a scheme.

To define an ensemble of \emph{symmetric} random matrices we start with a
`quadra\-tic' scheme of random variables $\left\{ Y_{N}(i,j)\right\}
_{(i,j)\in I_{N}}$ with $I_{N}=\left\{ (i,j)\,|\,1\leq i,j\leq N\right\} $
and define the matrix entries $X_{N}(i,j)$ by $X_{N}(i,j)=Y_{N}(i,j)$ for $%
i\leq j$ and $X_{N}(i,j)=Y_{N}(j,i)$ for $i>j$.
\end{definition}

\begin{remark}
To define the symmetric matrix $X_{N}$ it would be enough to start with a
`triangular' scheme of random variables, i.e.\ one with\goodbreak\noindent $I_{N}=\{
(i,j)\,|\,1\leq i\leq j\leq N\} $, thus with $M=\frac{1}{2}N(N+1)$ random
variables instead of $M=N^{2}$ variables. To reduce notational inconvenience
we decided to use the quadratic schemes. In a slight abuse of language we
will no longer distiguish in notation between the random variables $%
Y_{N}(i,j)$ and their symmetrized version $X_{N}(i,j)$. We will always
assume that the random matrices we are dealing with are symmetric.
\end{remark}

In this paper we consider schemes $\left\{ X_{N}(i,j)\right\} _{(i,j)\in
I_{N}}$ of random variables with $N=1,2,\ldots $and $I_{N}=\left\{
(i,j)\,|\,1\leq i,j\leq N\right\} $ with the following property:

\begin{definition}
\label{def:apcor} A scheme $\left\{ X_{N}(i,j)\right\} _{(i,j)\in I_{N}}$
is called \emph{%
approximately uncorrelated}, if
\begin{eqnarray}
\,\left\vert \mathbb{E}\left( \prod\limits_{\nu =1}^{\ell}X_{N}(i_{\nu },j_{\nu
})\,\,\prod\limits_{\rho =1}^{m}X_{N}(u_{\rho },v_{\rho })\right)
\right\vert \, &\leq &\ \frac{C_{\ell,m}}{N^{\ell/2}}  \label{eq:au1} \\
\left\vert \,\mathbb{E}\left( \prod\limits_{\nu =1}^{\ell}X_{N}(i_{\nu },j_{\nu
})^{2}\right) \,-\ 1\right\vert \ &\rightarrow &\ 0  \label{eq:au2}
\end{eqnarray}
for all sequences $(i_{1},j_{1}),(i_{2},j_{2}),\ldots ,(i_{\ell},j_{\ell})$ which
are pairwise disjoint and disjoint to the sequence $(u_{1},v_{1}),\ldots
,(u_{m},v_{m})$ with $N$-independent constants $C_{\ell,m}$.
\end{definition}
Note that for any approximately correlated scheme the mean asymptotically vanishes, $\left|\mathbb{E}(X_{N}((i,j))) \right| \leq C_{1,0} N^{-1/2} $ by \eqref{eq:au1}, and the variance is asymptotically one, $\mathbb{E}%
(X_{N}(i,j)^{2}) \to 1$ by \eqref{eq:au2}. Moreover, by \eqref{eq:au1}  we also have $\sup_{N,i,j}\mathbb{E}\left(
X_{N}(i,j)^{2k}\right) <\infty $ for all $k$.

The main examples we have in mind are schemes of Curie-Weiss- distributed
random variables (full or diagonal) with inverse temperature $\beta \leq 1$
(for details see Section \ref{sec:CW}).

\begin{theorem}
\label{th:Hres}If $\left\{ X_{N}(i,j)\right\} _{1\leq i\leq j\leq N}$ is an
approximately uncorrelated scheme of random variables then the
eigenvalue distribution measures $\sigma_{N}$  of the corresponding symmetric
matrices $A_{N}$ (as in (\ref{RMnorm})) converge weakly in probability
 to the semicircle
law $\sigma _{sc\,}$, i.e.\ for all bounded continuous functions $f$ on $%
\mathbb{R}
$ and all $ \varepsilon > 0 $ we have%
\begin{equation*}
\mathbb{P}\left(\left| \int \,f(x)\,d\sigma _{N}(x) - \int \,f(x)\,d\sigma
_{sc}(x)\right| > \varepsilon \right) ~\rightarrow~ 0 ~.
\end{equation*}
\end{theorem}
In particular, we prove the weak convergence of the density of states measure $ \overline{\sigma}_N $ to the semicircle law $\sigma _{sc\,}$.
In Section \ref{sec:CW} we discuss various examples of approximately
uncorrelated schemes.

\section{Proof of the semicircle law \label{sec:proof}}

The proof is a refinement of the classical moment method (see for example~\cite{AGZ}). We will sketch the proof emphasizing only the new ingredients.
As in \cite{AGZ}, Theorem~\ref{th:Hres} follows from the following two propositions.
\begin{proposition}\label{lem:Lemma6}
For all $ k \in \mathbb{N} $:
\begin{equation}
\frac{1}{N}\mathbb{E}\,\left( \mathrm{tr\,}A^{k}\right) \rightarrow \left\{
\begin{array}{cc}
C_{k/2} & \text{for }k\text{ even} \\
0 & \text{for }k\text{ odd}%
\end{array}%
\right.  \label{eq: moment}
\end{equation}
where $C_{k}=\frac{1}{k+1}\binom{2k}{k}$ denote the Catalan numbers.
\end{proposition}
The
right hand side of (\ref{eq: moment}) gives the moments of the semicircle
distribution $ \sigma_{sc} $.
In fact, this proposition implies the weak convergence of the density of states measures $ \overline{\sigma}_N $ to $ \sigma_{sc} $.
\begin{proposition}\label{lem:Lemma7}
For all $ k \in \mathbb{N} $:
\begin{equation}
\frac{1}{N^2}  \mathbb{E}\left[\left( \mathrm{tr\,}A^{k}\right)^2 \right]  \rightarrow \left\{
\begin{array}{cc}
C_{k/2}^2 & \text{for }k\text{ even} \\
0 & \text{for }k\text{ odd}%
\end{array}%
\right.   \, .
\end{equation}
\end{proposition}
\bigskip
Observe that Proposition \ref{lem:Lemma6} and Proposition \ref{lem:Lemma7} together imply that
\begin{equation}
 \mathbb{E}\left[\left( \frac{1}{N} \mathrm{tr\,}A^{k}\right)^2 \right] -  \mathbb{E}\left[\left( \frac{1}{N} \mathrm{tr\,}A^{k}\right) \right]^2 \to 0
\end{equation}
which allows us to conclude weak convergence in probability from weak convergence in the average (see \cite{AGZ}).

For a proof of the above propositions, which can be found in the subsequent subsections, we write%
\begin{eqnarray}
\frac{1}{N}\mathrm{tr\,}A^{k}~ &=&~\frac{1}{N^{1+k/2}}\,\sum_{i_{1},i_{2},%
\ldots ,i_{k}=1}^{N}\,X_{N}(i_{1},i_{2})\,X_{N}(i_{2},i_{3})\,\ldots
\,X_{N}(i_{k},i_{1})  \notag \\
&=&~\frac{1}{N^{1+k/2}}\,\sum_{i_{1},i_{2},\ldots ,i_{k}=1}^{N}\,X_{N}(%
\underline{i})\,  \label{eq:grSumme}
\end{eqnarray}
where we used the short hand notation
\begin{equation}
X_{N}(\underline{i})=X_{N}(i_{1},i_{2})\,X_{N}(i_{2},i_{3})\,\ldots
\,X_{N}(i_{k},i_{1})  \label{eq:XNi}
\end{equation}%
for \underline{$i$}$=(i_{1},i_{2},\ldots ,i_{k})$.
The associated $k+1 $-tupel $(i_{1},i_{2},\ldots ,i_{k},i_{1}) $ constitutes a Eulerian circuit
through the graph
$\mathcal{G}_{\underline{i}}$ (undirected, not necessary simple) with vertex set $%
V_{\underline{i}} =\left\{ i_{1},i_{2},\ldots ,i_{k}\right\} $ and an edge between the vertices $v$ and $w$
whenever $\left\{ v,w\right\} =\left\{ i_{j},i_{j+1}\right\} $ for some $%
j=1,\ldots k$ with the understanding that $i_{k+1}=i_{1}$, a convention we
keep for the rest of this paper.
More precisely, the number of edges $\nu
(v,w)$ linking the vertex $v$ and the vertex $w$ is given by
\begin{equation}
\nu (v,w)=\#\left\{ m\,|\,\left\{ v,w\right\} =\left\{ i_{m},i_{m+1}\right\}
\right\} ~~.  \label{eq:nuij}
\end{equation}
Let us call edges $e_{1}\neq e_{2}$ \emph{parallel} if they link the same
vertices. An edge which does not have a parallel edge is called \emph{simple}%
.
So, if $e$ links $v$ and $w$, then $e$ is a simple edge iff $\nu (v,w)=1$.
The graph $\mathcal{G}_{\underline{i}}$ may contain \emph{loops}, i.e. edges connecting a
vertex $v$ with itself. By a \emph{proper edge }we mean an edge which is not
a loop.
We set $\rho ($%
\underline{$i$}$)=\#\left\{ i_{1},i_{2},\ldots ,i_{k}\right\} $ the cardinality of the vertex set $ V_{\underline{i}}  $, i.e., the number of
(distinct) vertices the Eulerian circuit visits.
We also denote by $ \sigma(\underline{i}) $ the number of simple edges in the Eulerian circuit  $(i_{1},i_{2},\ldots ,i_{k},i_{1}) $.
With this notation we
can write (\ref{eq:grSumme}) as%
\begin{align}
\frac{1}{N}\mathrm{tr\,}A^{k}~& =~\,\frac{1}{N^{1+k/2}}\sum_{r=1}^{k}\sum_{%
\underline{i}:\,\rho (\underline{i})=r}\,X_{N}(\underline{i})   \notag \\
& =~ \frac{1}{N^{1+k/2}}\sum_{r=1}^{k}  \sum_{s=0}^k \sum_{%
\substack{ \rho(\underline{i})= r \\ \; \sigma(\underline{i}) = s   }}\,X_{N}(\underline{i})  .
\label{eq:grSumme2}
\end{align}

The sum extends over all Eulerian circuits with $k$
edges and vertex set $V_{\underline{i}}\subset \left\{ 1,2,\ldots ,N\right\} .$
To simplify future references we set

\begin{align}
S_{r,s}\,&=\,\sum_{\substack{ \rho(\underline{i})= r \\  \sigma(\underline{i}) = s   }}\left| \mathbb{E}\left[ X_N(\underline{i}) \right]\right| \\
\intertext{and}
S_{r}\,&=\,\sum_{\rho(\underline{i})= r}\left| \mathbb{E}\left[ X_N(\underline{i}) \right]\right|
\end{align}

{\setlength\parindent{0mm}
Obviously $\rho(\ui)$ and $\sigma(\ui)$ are integers with %
$1\leq\rho (\underline{i})\leq k$ and\goodbreak $ 0\leq\sigma(\underline{i})\leq k$. For $\rho(\underline{i})=r<N$ there are ${N\choose r}%
\le N^{r}$ choices for the  vertex set $V_{\underline{i}}$. Moreover,}
\begin{equation}\label{eq:equivc}
\# \left\{ \underline{i} \mid \rho(\underline{i} ) = r \right\} \leq \eta_k \, N^r
\end{equation}
where $ \eta_k $ is
the number of equivalence classes of Eulerian circuits of length~$k$.
We call two Eulerian circuits $ (i_{1},i_{2},\ldots ,i_{k},i_{1}) $ and $ (j_{1},j_{2},\ldots ,j_{k},j_{1}) $ with corresponding vertex sets $ V_{\underline{i}} $ and $ V_{\underline{j}} $
\emph{equivalent} if there is a bijection $ \varphi : V_{\underline{i}} \to V_{\underline{j}}  $ such that $ \varphi(i_m) = j_m $ for all $ m $.

\subsection{Proof of Proposition~\ref{lem:Lemma6}}
We investigate the expectation value of the sum~\eqref{eq:grSumme2}.
\begin{lemma}\label{lem:expa}
For all $ k \in \mathbb{N} $ there is some $ D_k < \infty $ such that for all $ N $:
\begin{equation}
S_{r,s}\,=\,\sum_{\substack{ \rho(\underline{i})= r \\ \; \sigma(\underline{i}) = s   }}\left| \mathbb{E}\left[ X_N(\underline{i}) \right]\right| \ \leq \ D_k \, N^{r- s/2} \, .
\end{equation}
\end{lemma}
\begin{proof}
The assertion follows
using \eqref{eq:au1} from the estimate $ \left| \mathbb{E}\left[X_N(\underline{i})\right]\right| \leq \widetilde D_k N^{-s/2} $ together with~\eqref{eq:equivc}.
\end{proof}

Evidently, in case $ r - s/2 < 1 + k/2 $ the term $\frac{1}{N^{k/2+1}}\,S_{r,s}$ vanishes in the limit. This is in particular the case if $ r < k/2 +1 $.
If $r> k/2+1$ we use the following proposition which is one of the key ideas of our
proof:

\begin{proposition}
\label{lem:graph}Let $\mathcal{G}=(V,E)$ denote a Eulerian graph with $r=\#V$
and $k=\#E$, and let $t$ be a positive integer such that $r\,>$ $\frac{k}{2}%
+t$ then $\mathcal{G}$ has at least $2t+1$ simple proper edges.
\end{proposition}

We note the following Corollary to Proposition~\ref{lem:graph}.

\begin{corollary}\label{cor:i}
For each $k$-tuple $\ui$ we have
\begin{equation*}
\rho(\ui)-\sigma(\ui)/2\leq k/2+1\;.
\end{equation*}
Moreover $\rho(\ui)-\sigma(\ui)/2=k/2+1 $
iff $\rho(\ui)=k/2+1 $ and $\sigma(\ui)=0 $.
\end{corollary}
\begin{proof}[Proof (Corollary \protect\ref{cor:i})]
Set $r=\rho(\ui)$ and $s=\sigma(\ui)$.\\
If $r \leq k/2+1$ the assertion is evident.\\
If $ r > k/2 + 1 $ there is some $ t \in \mathbb{N} $ such that
$$
\frac{k}{2} + t \ < \  r \ \leq \  \frac{k}{2} + t + 1 \, .
$$
Proposition~\ref{lem:graph} hence implies
$$
r - \frac{s}{2} \leq r - t - \frac{1}{2} \leq \frac{k}{2} + \frac{1}{2} < \frac{k}{2} + 1 \, .
$$
\end{proof}

Postponing the proof of Proposition~\ref{lem:graph}, we continue to prove Proposition~\ref{lem:Lemma6}.\\

From Lemma~\ref{lem:expa} and Corollary~\ref{cor:i} we learn that
\begin{equation}
\frac{1}{N^{k/2+1}}\,S_{r,s}~\to~0
\end{equation}
unless both $r=k/2+1$ and $s=0$.

Thus
it remains to compute the number of $k$-tuples $\underline{i},$ with $\rho (%
\underline{i})=1+k/2$ such that the corresponding graph is a \ `doubled'
planar tree.
There are $C_{\frac{k}{2}}$ different rooted planar trees with $\frac{k}{2}$
(simple) edges, where $C_{\ell }=\frac{1}{\ell +1}\binom{2\ell }{\ell }$ are
the Catalan numbers (see e.g. \cite{Stanley}, Exercise 6.19 e, p.\ 219-220).
Each index $i_{\nu }$ is chosen from the set $\left\{ 1,\ldots ,N\right\} $.
As we have $1+k/2$ different indices there are $\frac{N!}{(N-1-k/2)!}$ such
choices. From this one sees that%
\begin{eqnarray}
\lim_{N\rightarrow \infty }\frac{1}{N}\mathbb{E}\,\left( \mathrm{tr\,}%
A^{k}\right) &=&\lim_{N\rightarrow \infty }~\,\sum_{r=1}^{k}\,\frac{1}{%
N^{1+k/2}}\sum_{\rho (\underline{i})=r}\mathbb{E}\,\left(
X_{N}(\underline{i})\right)\label{eq:sum} \\
&=&\lim_{N\rightarrow \infty }~\,\,\frac{1}{N^{1+k/2}}\sum_{\rho (\underline{i})=1+k/2, \sigma(\underline{i}) = 0 }\mathbb{E}\,\left( X_{N}(\underline{i})\right)\notag
\\
&=&\lim_{N\rightarrow \infty }~\,\,\frac{1}{N^{1+k/2}}~\frac{N!}{(N-1-k/2)!}%
~C_{\frac{k}{2}} \, = ~C_{\frac{k}{2}}\notag .
\end{eqnarray}
This  ends the proof of Lemma~\ref{lem:Lemma6} modulo the proof of the Proposition \ref{lem:graph}.
For future purpose, we note that the above proof also shows the slightly stronger assertion.
\begin{corollary}\label{cor:cor9}
For all $ k \in \mathbb{N} $:
\begin{equation}\label{eq:sumEC}
\lim_{N \to \infty} \frac{1}{N^{1+k/2}} \sum_{\underline{i}} \left| \mathbb{E}\left[X_N(\underline{i})\right]\right| = \left\{ \begin{array}{cc}
C_{k/2} & \text{for }k\text{ even} \\
0 & \text{for }k\text{ odd}%
\end{array}%
\right.   \,
\end{equation}
where the above sum (in \eqref{eq:sumEC}) extends over all Eulerian circuits of length $ k $.
\end{corollary}

For a proof we note that the leading contribution in the sum \eqref{eq:sum} is non-negative. The subleading terms were already shown to vanish.

\bigskip
\begin{proof}[Proof (Proposition \protect\ref{lem:graph})]
If the graph $\cG=(V,E)$ contains loops, we
delete all loops and call the new graph $(V,\tilde{E})$. This graph is still
Eulerian and satisfies $\#V > \#\tilde{E}/2+t.$ Thus without loss of
generality we may assume that $(V,E)$ contains no loops.

We proceed by induction on the number of edges with multiplicity greater
than one. If there is no such edge, then the number of simple edges is $k$.
Since $\cG$ is Eulerian we have $k\geq r$ and $r>\frac{k%
}{2}+t$ implies $k>2t$ and thus the assertion.

Hence, assume there exists an edge of multiplicity $m\geq 2$. If the graph
that arises from the deletion of $2$ copies of this edge is still connected,
then the resulting graph is Eulerian and denoting its number of edges by $%
k^{\prime }$ we have
\begin{equation*}
r>\frac{k}{2}+t\geq \frac{k^{\prime }}{2}+t+1
\end{equation*}

Thus, by inductive assumption, we find at least $2t+3$ edges without
parallels in the reduced graph and hence at least $2t+2$ in $\cG$.

We are left with the case that the removal of the edges disconnects the
graph into two Eulerian graphs $\cG_{1}=(V_{1},E_{1})$ and $%
\cG_{2}=(V_{2},E_{2}) $. We use the abbreviations $r_{i}:=\#V_{i}$ and $%
k_{i}:=\#E_{i}$. Then:%
\begin{equation*}
r=r_{1}+r_{2}>\frac{k}{2}+t=\frac{k_{1}+1}{2}+\frac{k_{2}+1}{2}+t
\end{equation*}

Hence we can partition $t$ into integers $t_{1},t_{2}$ such that
\begin{equation*}
r_{1}>\frac{k_{1}}{2}+t_{1}\text{ \ and \ }r_{2}>\frac{k_{2}}{2}+t_{2}
\end{equation*}

If, say $t_{1}\leq 0$, then $t_{2}\geq t$ and the inductive assumption
yields at least $2t+1$ simple edges in $\cG_{2}$ and hence in $\cG$.

Otherwise we find at least $2t_{i}+1$ simple edges in each of the
$\cG_{i}$ and thus in total $2t+2>2t+1$ such edges in $\cG$.
\end{proof}

\subsection{Proof of Proposition~\ref{lem:Lemma7}}

We write the expectation value
\begin{equation}\label{eq:double}
\frac{1}{N^2} \mathbb{E}\left[ \left( \mathrm{tr\,}A^k\right)^2 \right] = \frac{1}{N^{2+k} } \sum_{\underline{i}, \underline{j}} \mathbb{E}\left[ X_N(\underline{i}) \, X_N( \underline{j}) \right]
\end{equation}
where the sum extends over all pairs of Eulerian circuits $ (i_1, \dots , i_k, i_1) $ and $  (j_1, \dots , j_k, j_1) $ of length $k$ with vertex sets $ V_{\underline{i}} $ and $ V_{\underline{j}} $ in $ \{ 1, \dots , N \} $. We distinguish two cases.

In case $ V_{\underline{i}} \cap V_{\underline{j}} \neq \emptyset $ the union of the corresponding Eulerian graphs $ \mathcal{G}_{\underline{i}} \cup \mathcal{G}_{\underline{j}} $ is connected and each vertex has even degree. Therefore this union is itself a Eulerian graph with $ 2k $ edges. The corresponding contribution to the sum~\eqref{eq:double} is then estimated by extending the summation to all Eulerian circuits $ \underline{\ell} $ of length $ 2k $:
\begin{equation}
\frac{1}{N^{2+k} } \sum_{\substack{\underline{i}, \underline{j} \\ V_{\underline{i}} \cap V_{\underline{j}} \neq \emptyset} } \left| \mathbb{E}\left[ X_N(\underline{i}) \, X_N( \underline{j}) \right] \right| \leq \frac{1}{N^{2+k} } \mkern-7mu \sum_{\underline{\ell} = (\ell_1, \dots, \ell_{2k})} \mkern-5mu\left| \mathbb{E}\left[X_N(\underline{\ell}) \right] \right| \leq \frac{\widetilde{C}_k}{N} \, .
\end{equation}
The last estimate is due to Corollary~\ref{cor:cor9}.

In case $ V_{\underline{i}} \cap V_{\underline{j}} = \emptyset $ we use the following analogue of Lemma~\ref{lem:expa}.
\begin{lemma}
For all $ k \in \mathbb{N} $ there is some $ D_k < \infty $ such that for all $ N $:
\begin{equation}
\sum_{\substack{ V_{\underline{i}} \cap V_{\underline{j}} = \emptyset  \\ \rho(\underline{i})= r_1  , \, \rho(\underline{j}) =r_2 \\ \; \sigma(\underline{i}) = s_1 , \, \sigma(\underline{j}) = s_2   }}\left| \mathbb{E}\left[ X_N(\underline{i}) \, X_N(\underline{j}) \right]\right| \ \leq \ D_k \, N^{r_1+r_2- s_1/2- s_2/2} \, ,
\end{equation}
where the sum extends over non-intersecting pairs of Eulerian circuits of length $ k $.
\end{lemma}
The proof mirrors that of Lemma~\ref{lem:expa}.

From Lemma~\ref{lem:expa} we know that $ r_1 - s_1/2 \leq k/2 + 1 $ and likewise $ r_2 - s_2/2 \leq k/2 + 1 $. So the unique possibility that
$$
r_1 + r_2 - \frac{s_1+s_2}{2} \geq k + 2
$$
giving rise to a non-vanishing term in the limit, is that $ r_1 = r_2 = k/2 + 1 $ and $ s_1 = s_2 =0 $. Similarly as in the proof of Lemma~\ref{lem:Lemma6} we conclude that in this case $ \underline{i} $ and $ \underline{j} $ constitute disjoint 'doubled' planar trees and $   \mathbb{E}\left[ X_N(\underline{i}) \, X_N(\underline{j}) \right]\to 1 $ by assumption~\eqref{eq:au2}.
The proof of Lemma~\ref{lem:Lemma7} is concluded using the same arguments relating  the number
of planar trees to the Catalan numbers.

\section{The Curie-Weiss model and its relatives \label{sec:CW}}

In this section we discuss the Curie-Weiss model and related ensembles in the
framework of general exchangeable sequences. Let us first recall:

\begin{definition}
A finite sequence $X_{1},X_{2},\ldots ,X_{M}$ of random variables is called
\emph{exchangeable} if for any permutation $\pi \in \mathcal{S}_{M}$ the joint
distributions of $X_{1},X_{2},\dots , X_{M}$ and of $X_{\pi (1)},X_{\pi
(2)},\dots ,X_{\pi (M)}$ agree. An infinite sequence $\left\{ X_{i}\right\}
_{i\in I}$ is called exchangeable if any finite subsequence is.
\end{definition}

It is a well known result by de Finetti (\cite{de Finetti}, for further
developments see e.g.~\cite{Aldous}) that any exchangeable sequence of $%
\left\{ -1,1\right\} $-valued random variables is a mixture of independent
random variables. To give this informal description a precise meaning we
define:

\begin{definition}
For $t\in \left[ -1,1\right] $ we denote by $P_{t}$ the probability measure\goodbreak\noindent $%
P_{t}=\frac{1}{2}(1+t)\,\delta _{1}\,+\,\frac{1}{2}(1-t)\,\delta _{-1}$ on $%
\left\{ -1,1\right\} $, i.e. $P_{t}(1)=\frac{1}{2}(1+t)$ and $P_{t}(-1)=%
\frac{1}{2}(1-t)$. By $P_{t}^{M}=P_{t}^{\otimes _{M}}$ we mean the $M$-fold,
by $P_{t}^{\infty }=P_{t}^{\otimes _{\mathbb{N}}}$ the infinite product of
this measure.
\end{definition}

\begin{remark}
The measures $P_{t}$ are parametrized in such a way that $E_{t}(X):=\int
x\,dP_{t}=t$.  To simplify notation, we write $P_{t}^{M}(x_{1},x_{2},\dots
,x_{M})$ instead of $P_{t}^{M}(\{(x_{1},x_{2},\dots ,x_{M})\})$.
\end{remark}

We are now in a position to formulate de Finetti's theorem:

\begin{theorem}[de Finetti]
If\, $\left\{ X_{i}\right\} _{i\in \mathbf{N}}$ is an exchangeable sequence of\\
$\left\{ -1,1\right\} $-valued random variables with distribution $\mathbb{P}
$ (on $\left\{ -1,1\right\} ^{\mathbb{N}})$ then there exists a probability
measure $\mu $ on $\left[ -1,1\right] $, such that for any measurable set $%
S\subset \left\{ -1,1\right\} ^{\mathbb{N}}$:%
\begin{equation*}
\mathbb{P}(S)=\int \,P_{t}^{\infty }(S)\,d\mu (t)~.
\end{equation*}
\end{theorem}

For this result it is essential that the index set $I=\mathbb{N}$ is
infinite. In fact, the theorem does not hold for finite sequences in general
(see e. g. \cite{Aldous}).

\begin{definition}
If $\mu $ is a probability measure on $\left[ -1,1\right] $ then we call a
measure%
\begin{equation}
\mathbb{P(\cdot )}=\int \,P_{t}^{M}(\cdot )\,d\mu (t)
\end{equation}
on $\left\{ -1,1\right\} ^{M}$ a measure of \emph{de Finetti type} (with de
Finetti measure $\mu $). We say that a finite sequence $\left\{
X_{1},X_{2},\ldots ,X_{M}\right\} $ of random variables is of de Finetti
type if the joint distribution of the $\left\{ X_{i}\right\} _{i=1}^{M}$ is
of de Finetti type.
\end{definition}

The following observation allows us to compute correlations of de Finetti
type random variables:

\begin{proposition}
\label{prop:moments}If the sequence $\left\{ X_{1},X_{2},\ldots
,X_{M}\right\} $ of random variables is of de Finetti type with de Finetti
measure $\mu $ then for distinct $i_{1},\ldots ,i_{K}$%
\begin{equation*}
\mathbb{E}(X_{i_{1}}\,X_{i_{2}}\,\ldots \,X_{i_{K}})~=~\int \,t^{K}~d\mu (t)\ .
\end{equation*}
\end{proposition}

\begin{proof}
By the definition of $P_{t}^{M}$ we have $E_{t}^{M}(X_{i_{1}}\,X_{i_{2}}\,%
\ldots \,X_{i_{K}})=t^{K}$.
\end{proof}

\begin{corollary}
Suppose $\mathbb{P}_{N}\mathbb{(\cdot )}=\int \,P_{t}^{N^{2}}(\cdot )\,d\mu
_{N}(t)$ is a sequence of measures of de Finetti type and $X_{N}$ is a
random matrix ensemble corresponding to $\mathbb{P}_{N}$ via Definition \ref%
{def:SchemeMatrix}. If for all $k\in
\mathbb{N}
$%
\begin{equation}
\int \,t^{K}~d\mu _{N}(t)~\leq ~\frac{C_{K}}{N^{K/2}}  \label{eq.expT}
\end{equation}
for some constants $C_K$, then $X_{N}$ satisfies the semicircle law.
\end{corollary}

\begin{proof}
We prove that $\left\{ X_{N}(i,j)\right\} $ is approximately uncorrelated in
the sense of Definition \ref{def:apcor}. Since $X_{N}(i,j)^{2}=1$ property (%
\ref{eq:au2}) is evident. Property (\ref{eq:au1}) follows from (\ref{eq.expT}%
) and Proposition \ref{prop:moments}.
\end{proof}\\

Curie-Weiss distributed random variables turn out to be examples of de
Finetti sequences. This fact is contained in a somewhat hidden way in
physics textbooks (see for example \cite[section 4-5]{Thompson}).

\begin{theorem}
\label{th:CWrepr}Curie-Weiss ($\,\mathbb{P}_{\beta ,M}$-)\,distributed random
variables\goodbreak\noindent $\left\{ X_{1},X_{2},\ldots ,X_{M}\right\} $ are of\ de Finetti
type, more precisely%
\begin{eqnarray*}
&&\mathbb{P}_{\beta ,M}(X_{1}=x_{1},X_{2}=x_{2},\ldots ,X_{M}=x_{M}) \\
&=&\ Z^{-1}\int_{-1}^{+1}\,P_{t}^{M}(x_{1},x_{2},\ldots ,x_{M})\,\frac{%
e^{-MF_{\beta }(t)\,/2}}{1-t^{2}}\,dt
\end{eqnarray*}
where $F_{\beta }(t)=\frac{1}{\beta }\left( \frac{1}{2}\ln \frac{1+t}{1-t}%
\right) ^{2}+\ln \left( 1-t^{2}\right) $ and the normalization factor is
given by $Z=\int \frac{e^{-MF_{\beta }(t)\,/2}}{1-t^{2}}\,dt$.
\end{theorem}

\begin{proof}
Using the observation $e^{\frac{z^{2}}{2}}=(2\pi )^{-\frac{1}{2}%
}\int_{-\infty }^{+\infty }\,e^{-\frac{s^{2}}{2}+sz}\,ds$ (also known as
Hubbard-Stratonovich transformation) we obtain%
\begin{eqnarray*}
&&\mathbb{P}_{\beta ,M}(X_{1}=x_{1},X_{2}=x_{2},\ldots ,X_{M}=x_{M})\, \\
&=&\,Z_{\beta ,M}^{-1}\,\frac{1}{2^{M}}\,e^{\frac{\beta }{2M}(\sum
x_{j})^{2}} \\
&=&\,(2\pi )^{-1/2}\,Z_{\beta ,M}^{-1}\,\frac{1}{2^{M}}\int_{-\infty
}^{+\infty }\,e^{-\frac{s^{2}}{2}+s\sqrt{\frac{\beta }{M}}\sum x_{j}}\;\,ds \\
~ &&~~\text{setting }y=\sqrt{\frac{\beta }{M}}\,s\,\text{ we obtain:} \\
&=&(2\pi )^{-1/2}\,Z_{\beta ,M}^{-1}\,\sqrt{\frac{M}{\beta }}\int_{-\infty
}^{+\infty }e^{-\frac{M}{2\beta }y^{2}}\cosh ^{M}y\,\,\left( \frac{1}{%
2^{M}\cosh ^{M}y}\prod_{i=1}^{M}\,e^{yx_{i}}\right) \;dy \\
&=&(2\pi )^{-1/2}\,Z_{\beta ,M}^{-1}\,\sqrt{\frac{M}{\beta }}\int_{-\infty
}^{+\infty }e^{-M(\frac{y^{2}}{2\beta }-\ln \cosh
y)}\,\,\prod\limits_{i=1}^{M}\left( \frac{e^{yx_{i}}}{\cosh y}%
\,P_{0}(x_{i})\right) \;dy \\
&&\text{a change }t=\tanh y\text{ of variables gives:} \\
&=&(2\pi )^{-1/2}\,Z_{\beta ,M}^{-1}\,\sqrt{\frac{M}{\beta }}%
\int_{-1}^{+1}e^{-M\,F_{\beta }(t)\,/2}\,P_{t}(x_{1},x_{2},\ldots ,x_{M})~%
\frac{1}{1-t^{2}}~dt\, \\
&=&Z^{-1}~\int_{-1}^{+1}\,P_{t}^{M}(x_{1},x_{2},\ldots ,x_{M})\,\frac{%
e^{-MF_{\beta }(t)\,/2}}{1-t^{2}}\,dt~~~\text{.}
\end{eqnarray*}

Above we used that for $|\,t\,|<1$ we have $\tanh ^{-1}(t)=\frac{1}{2}\ln
\frac{1+t}{1-t}$, $\frac{dt}{dy}=\frac{1}{\cosh ^{2}y}=\frac{\cosh
^{2}y-\sinh ^{2}y}{\cosh ^{2}y}=1-\tanh ^{2}y$, and $\ln \cosh y=-\frac{1}{2}%
\ln (1-\tanh ^{2}y)$.
\end{proof}

\begin{remark}
From the above proof an alternative representation of the Curie-Weiss
probability follows. Defining the measure $Q_{y}=\frac{1}{2\cosh y}%
(e^{y}\delta _{1}+e^{-y}\delta _{-1})$ and $Q_{y}^{M}$ its $M$-fold product
we may write%
\begin{equation*}
\mathbb{P}_{\beta ,M}(x_{1},x_{2},\ldots ,x_{M})=\widetilde{Z}%
^{-1}\,\int_{-\infty }^{+\infty }\,e^{-M(\frac{y^{2}}{2\beta }+\ln \cosh
y)}\,Q_{y}^{M}(x_{1},x_{2},\ldots ,x_{M})\,dy~~\ \text{.}
\end{equation*}

This formula occurs in the physics literature (at least in disguise).
\end{remark}

\begin{definition}
\label{def:PPhiM} Let $F:\left( -1,1\right) \rightarrow
\mathbb{R}
$ be a measurable function such that $Z=\int_{-1}^{1}\frac{e^{-N\,F(t)\,/2}}{%
1-t^{2}}\,dt$ is finite for all $N\in
\mathbb{N}
$, then the probability measure $\mathbb{P}_{M}^{N\,F}$ on $\left\{
-1,1\right\} ^{M}$ is defined by%
\begin{equation}
\mathbb{P}_{M}^{N\,F}(x_{1},x_{2},\ldots
,x_{M})=Z^{-1}\,\int_{-1}^{+1}\,P_{t}^{M}(x_{1},x_{2},\ldots ,x_{M})\,\frac{%
e^{-N\,F(t)\,/2}}{1-t^{2}}\,dt~\ \text{.}  \label{eq:PPhiM}
\end{equation}

We call a measure of the form $\mathbb{P}_{M}^{N\,F}$ a \emph{generalized
Curie-Weiss} measure.
\end{definition}

\begin{remark}
Obviously, $\mathbb{P}_{M}^{N\,F}$ is a measure of de Finetti type and we
have $\mathbb{P}_{\beta ,M}=\mathbb{P}_{M}^{M\cdot F_{\beta }}$. Note that $%
N $ and $M$ may be different in general.
\end{remark}

The advantage of the form (\ref{eq:PPhiM}) is that for many cases we can
compute the asymptotics of the correlation functions as $N\rightarrow \infty
$ using the Laplace method:

\begin{proposition}[Laplace method \protect\cite{Olver}]
\label{prop:Laplace}Suppose ${\large F:}\left( -1,1\right) \rightarrow
\mathbb{R}
$ is differentiable and $\phi :\left( -1,1\right) \rightarrow
\mathbb{R}
$ is measurable and for some $a\in(-1,1)$ we have

\begin{enumerate}
\item $\inf_{x\in \left[ a,1\right] }\,F(x)=F(a)$ and $\inf_{x\in \left[
b,1\right] }F(x)>F(a)$ for all $b\in \left( a,1\right) $.

\item $F^{\prime }$ and $\phi $ are continuous in a neighborhood of $a$.

\item As $x\searrow a$ we have%
\begin{eqnarray}
F(x) &=&F(a)+P\,\left( x-a\right) ^{\nu }\,+\,\mathcal{O(}\left( x-a\right)
^{\nu +1})  \label{eq:Olver1} \\
\phi (x) &=&Q\,\left( x-a\right) ^{\lambda -1}\,+\,\mathcal{O(}\left(
x-a\right) ^{\lambda })  \label{eq:Olver2}
\end{eqnarray}

where $\nu ,\lambda $ and $P$ are positive constants and $Q$ is a real
constant and (\ref{eq:Olver1}) is differentiable.

\item The integral $I(N)=\int_{a}^{1}e^{-N\,F(x)\,/2}\,\,\phi \left(
x\right) \,dx$ is finite for all sufficiently large $N.$
\end{enumerate}

Then as $N\rightarrow \infty $%
\begin{equation*}
I\left( N\right) ~\approx ~\frac{Q}{\nu }\;\;\Gamma\!\left( \frac{\lambda }{\nu
}\right) \,P^{-\frac{\lambda }{\nu }}\,\left( \frac{N}{2}\right) ^{-\frac{%
\lambda }{\nu }}\,e^{-N\,F(a)/2}
\end{equation*}
where $A(N)\approx B(N)$ means $\lim_{N\to\infty} \frac{A(N)}{B(N)}=1$ and\, $\Gamma$ denotes the Gamma function..
\end{proposition}

\begin{remark}
This theorem and its proof can be found in \cite[Ch. 3 \S 7]{Olver}.
\end{remark}

We apply the Laplace method to a few interesting cases of $\mathbb{P}%
_{M}^{N\,F}$.

\begin{theorem}
\label{th:deFiCorr} Let $F:\left( -1,1\right) \rightarrow
\mathbb{R}
$ be a smooth even function with $F(t)\rightarrow \infty $ as $t\rightarrow
\pm 1$ such that $\int_{-1}^{1}e^{-N\,F(t)\,/2}\,\,t^{p}\,\frac{dt}{1-t^{2}}$ is finite for
all $p\geq 0$ and all $N$ big enough and suppose that $F$ has a unique
minimum in $[0,1)$ at $t=a$. Then we have for distinct $X_{1},\,X_{2},\,%
\ldots \,,X_{K}$,  $K\leq M$ as $N\to\infty$ and uniformly in $M$:

\begin{enumerate}
\item \label{case:1}If $a=0$ and $F^{\prime \prime }(0)>0$ (i. e. $F$ has a
quadratic minimum at $0$), then%
\begin{eqnarray*}
\text{for }K\text{ even:}\hspace{1.5cm}&& \\ \mathbb{E}_{M}^{N\,F}(X_{1}\,X_{2}\,\ldots
\,X_{K})~ &\approx &~(k-1)!!\,\ \frac{1}{(\frac{1}{2}F^{\prime \prime
}(0))^{K/2}}\,\frac{1}{N^{K/2}} \\
\text{and for }K\text{ odd :\ }\hspace{1cm}&&\\[2mm] \mathbb{E}_{M}^{N\,F}(X_{1}\,X_{2}\,\ldots
\,X_{K})~ &=&~0~~\text{.}
\end{eqnarray*}

\item \label{case:2}If $a=0$ and $F^{\prime \prime }(0)=0$, $F^{(4)}(0)>0$
(i. e. $F$ has a quartic minimum at $0$), then%
\begin{eqnarray*}
\text{for }K\text{ even: }\hspace{1.5cm}&&\\ \mathbb{E}_{M}^{N\,F}(X_{1}\,X_{2}\,\ldots
\,X_{K})~ &\approx &~C_{K}\,\ \frac{1}{(\frac{1}{24}F^{(4)}(0))^{K/4}}\,%
\frac{1}{N^{K/4}} \\
\text{and for }K\text{ odd: \ }\hspace{1cm}&&\\[2mm] \mathbb{E}_{M}^{N\,F}(X_{1}\,X_{2}\,\ldots
\,X_{K})~ &=&~0
\end{eqnarray*}

where $C_{K}=\frac{\Gamma (\frac{k+1}{4})}{\Gamma (\frac{1}{4})}2^{K/4}$.

\item \label{case:3}If $a>0$ and $F^{\prime \prime }(a)>0$ then%
\begin{equation*}
\mathbb{E}_{M}^{N\,F}(X_{1}\,X_{2}\,\ldots \,X_{K})~\approx ~\frac{1}{2}%
\,\left( a^{K}+a^{-K}\right)~.
\end{equation*}
\end{enumerate}
\end{theorem}

\begin{proof}
The proof of Theorem \ref{th:deFiCorr} relies on the Laplace method
(Proposition~\ref{prop:Laplace}). We concentrate on the proof of case \ref%
{case:1}, the other cases are proved by the same reasoning.

We set
\begin{equation*}
Z_{K}=\int_{-1}^{+1}\,e^{-N\,F(t)\,/2}\,\frac{t^{K}}{1-t^{2}}\,dt~~\text{.}
\end{equation*}

Then by (\ref{eq:PPhiM}) and Proposition \ref{prop:moments} we have $\mathbb{%
E}_{M}^{N\,F}(X_{1}\,X_{2}\,\ldots \,X_{K})=\frac{Z_{K}}{Z_{0}}$. For $K$
odd we have $Z_{K}=0$ since $\phi (t)=t^{K}$ is odd in this case. For even $%
K $ we have $Z_{K}=2\,\tilde{Z}_{K}$ with $\tilde{Z}_{K}=\int_{0}^{+1}%
\,e^{-N\,F(t)\,/2}\,\frac{t^{K}}{1-t^{2}}\,dt$. Moreover, $F(t)=F(0)+\frac{%
t^{2}}{2}F^{\prime \prime }(0)+\mathcal{O}(t^{3})$. Applying Proposition \ref%
{prop:Laplace} both to $\tilde{Z}_{K}$ and to $\tilde{Z}_{0}$ we obtain:%
\begin{eqnarray*}
\tilde{Z}_{K}~ &\approx &~\Gamma \left( \frac{k+1}{2}\right) \,\left( \frac{1%
}{\frac{1}{2}F^{\prime \prime }(0)}\right) ^{\frac{K+1}{2}}\,\left( \frac{2}{%
N}\right) ^{\frac{K+1}{2}}\,e^{-N\,F(0)} \\
\tilde{Z}_{0}~ &\approx &~\Gamma \left( \frac{1}{2}\right) \,\left( \frac{1}{%
\frac{1}{2}F^{\prime \prime }(0)}\right) ^{\frac{1}{2}}\,\left( \frac{2}{N}%
\right) ^{\frac{1}{2}}\,e^{-N\,F(0)}~~\text{.}
\end{eqnarray*}

Hence, we get%
\begin{eqnarray*}
\mathbb{E}_{M}^{N\,F}(X_{1}\,X_{2}\,\ldots \,X_{K})~ &=&~\frac{\tilde{Z}_{K}%
}{\tilde{Z}_{0}} \\
&\approx &\frac{\Gamma \left( \frac{k+1}{2}\right) }{\Gamma \left( \frac{1}{2%
}\right) }\,2^{K/2}\,\,\left( \frac{1}{\frac{1}{2}F^{\prime \prime }(0)}%
\right) ^{\frac{K}{2}}\,\left( \frac{1}{N}\right) ^{\frac{K}{2}}~.
\end{eqnarray*}

The result (\ref{case:1}) then follows from the observation that $\frac{%
\Gamma \left( \frac{k+1}{2}\right) }{\Gamma \left( \frac{1}{2}\right) }%
\,2^{K/2}=(K-1)!!$ for even $K.$ Case \ref{case:2} can be handled in a
similar way.

For case \ref{case:3} we note that $-a$ is also a minimum of the function $F$
since $F$ is even. We devide the integral $\int_{-1}^{1}$ into four parts,
namely $\int_{-1}^{-a}\,+\,\int_{-a}^{0}\,+\int_{0}^{a}\,+\,\int_{a}^{1}$
and observe that each of these terms has the same asymptotics as $%
N\rightarrow \infty $.
\end{proof}

\begin{remark}
\label{rem:ExpTbet}\bigskip As a remark to the above proof we notice that
under the assumptions in case \ref{case:1} we have%
\begin{equation}
\left( \int_{-1}^{+1}\,e^{-\frac{N\,F(t)}{2}}\,\frac{dt}{1-t^{2}}\,\right)
^{-1}\int_{-1}^{+1}\!\!\! e^{- \frac{N\,F(t)}{2}}\,\frac{|t|}{1-t^{2}}%
\,dt~\approx ~\frac{\sqrt{2}}{\sqrt{\pi }}\frac{1}{\sqrt{\frac{1}{2}%
F^{\prime \prime }(0)}}\frac{1}{\sqrt{N}}  \label{eq:Exptbet}
\end{equation}
and in case \ref{case:2} we obtain%
\begin{equation*}
\left( \int_{-1}^{+1}\,e^{-\frac{N\,F(t)}{2}}\,\frac{dt}{1-t^{2}}\,\right)
^{-1}~\int_{-1}^{+1}\,e^{-\frac{N\,F(t)}{2}}\,\frac{|t|}{1-t^{2}}%
\,dt~\approx ~C_{\beta }\frac{1}{N^{1/4}}~~.
\end{equation*}
\end{remark}

\begin{corollary}
\label{cor:CWCorr}Let \ $F_{\beta }(t)=\frac{1}{\beta }\left( \frac{1}{2}\ln
\frac{1+t}{1-t}\right) ^{2}+\ln \left( 1-t^{2}\right) $ and let $M(N)$ be a
function of $N$ and $K\leq M(N)$ for $N$ large enough and let $%
X_{1},\,X_{2},\,\ldots \,,X_{K}$ be a sequence of distinct random variables.
As before we set

\begin{equation*}
\mathbb{\mathbb{E}}_{M(N)}^{N\,F_{\beta }}(\cdot
)=Z^{-1}\,\int_{-1}^{+1}\,E_{t}^{M(N)}(\cdot )\,\frac{e^{-\frac{N\,F(t)}{2}}%
}{1-t^{2}}\,dt~\ \text{.}
\end{equation*}

\begin{enumerate}
\item For $\beta <1$ we have%
\begin{eqnarray*}
\text{for }K\text{ even: }&&\\
\mathbb{\mathbb{E}}_{M(N)}^{N\,F_{\beta
}}(X_{1}\,X_{2}\,\ldots \,X_{K})~ &\approx &~(k-1)!!\,\ \left( \frac{\beta }{%
1-\beta }\right) ^{K/2}\,\frac{1}{N^{K/2}} \\
\text{for }K\text{ odd : }&&\\
\mathbb{\mathbb{E}}_{M(N)}^{N\,F_{\beta
}}(X_{1}\,X_{2}\,\ldots \,X_{K})~ &=&~0~~\text{.}
\end{eqnarray*}

\item For $\beta =1$ we have for a constant $c_{K}>0$:
\begin{eqnarray*}
\text{for }K\text{ even: }\qquad\mathbb{\mathbb{E}}_{M(N)}^{N\,F_{1
}}(X_{1}\,X_{2}\,\ldots \,X_{K})~ &\approx &~c_{K}\,\ \,\frac{1}{N^{K/4}} \\
\text{for }K\text{ odd: \ }\qquad\mathbb{\mathbb{E}}_{M(N)}^{N\,F_{1}}(X_{1}\,X_{2}\,\ldots \,X_{K})~ &=&~0~~\text{.}
\end{eqnarray*}

\item For $\beta >1$we have%
\begin{equation}
\mathbb{\mathbb{E}}_{M(N)}^{N\,F_{\beta }}(X_{1}\,X_{2}\,\ldots
\,X_{K})~\approx ~\frac{1}{2}\left( m(\beta )^{K}+(-m(\beta ))^{K}\right)
\label{eq:mvonbeta}
\end{equation}

where $m(\beta )>0$ is the unique positive solution of $\;\tanh(\beta t)\,=\,t$.
\end{enumerate}
\end{corollary}

\begin{proof}
Let us compute the minima of the function $F_{\beta }.$ We have:%
\begin{equation*}
F_{\beta }^{\prime }(t)~=~\frac{1}{1-t^{2}}\left( \frac{1}{\beta }\,\ln
\frac{1+t}{1-t}-2\,t\right)
\end{equation*}
hence the possible extrema $m$ of $F_{\beta }$ satisfy:%
\begin{equation*}
\frac{1}{2}\,\ln \frac{1+m}{1-m}=\beta \,m
\end{equation*}
or equivalently%
\begin{equation*}
\tanh \beta m=m~\text{.}
\end{equation*}

For $\beta <1$ the only solution is $m=0$ and this solution is a quadratic minimum
since $F_{\beta }^{\prime \prime }(0)=2\,\frac{1-\beta }{\beta }>0$ for $%
\beta <1$.

For $\beta =1$ the solution $m=0$ is a quartic minimums as $F_{1}^{\prime\prime}(0)=0$ and $F_{1}^{(4)}(0)=4$.

For $\beta >1$
the solution $m=0$ is a maximum of $F_{\beta }$ and there is a positive
solution $m$ which is a minimum. The same is true for $-m$.

With this information we can apply Theorem \ref{th:deFiCorr}.
\end{proof}\\

Now, we discuss random matrix ensembles defined through generalized
Curie-Weiss models.

\begin{definition}
Suppose $\alpha>0$ and $F:\left( -1,1\right) \rightarrow
\mathbb{R}
$ is a smooth even function with $F(t)\rightarrow \infty $ as $t\rightarrow
\pm 1$ and such that $\int_{-1}^{1}e^{-N^{\alpha }\,F(t)\,/2}\,\,t^{p}\,\frac{dt}{1-t^{2}}$
is finite for all $p\geq 0$ and all $N$ big enough. Let $\left\{
Y_{N}(i,j)\right\} _{1\leq i,j\leq N}$ be a quadratic scheme of\, $\mathbb{P}%
_{N^{2}}^{N^{\alpha }\,F}$-distributed random variables, and set $%
X_{N}(i,j)=Y_{N}(i,j)$ for $i\leq j$ and $X_{N}(i,j)=Y_{N}(j,i)$ for $i>j$.
Then we call the random matrix ensemble $X_{N}(i,j)$ a generalized ($\mathbb{%
P}_{N^{2}}^{N^{\alpha }\,F}$)-Curie-Weiss ensemble.
\end{definition}

\begin{remark}
The full Curie-Weiss ensemble is a $\mathbb{P}_{N^{2}}^{N^{2}\,F_{\beta }}$%
-ensemble.
\end{remark}

\begin{theorem}
Suppose the random matrix ensemble $X_{N}(i,j)$ is a generalized\, $\mathbb{P}%
_{N^{2}}^{N^{\alpha }\,F}$-Curie-Weiss ensemble.

\begin{enumerate}
\item If $F$ has a unique quadratic minimum at $a=0$ and $\alpha \geq 1$
then the semicircle law holds for $X_{N}$.

\item If $F$ has a unique quartic minimum at $a=0$ and $\alpha \geq 2$ then
the semicircle law holds for $X_{N}$.
\end{enumerate}
\end{theorem}

\section{Largest eigenvalue \label{sec:largeig}}

At a first glance one might expect that for matrix ensembles with
generalized Curie-Weiss distribution the limit density of states measure $\mu$ should
depend on $\beta $, even for $\beta \leq 1$. After all, the correlation
structure of the ensemble depends strongly on $\beta $: the behavior of the
covariance is given by $\mathbb{E}_{\beta ,M}(X_{1}X_{2})\approx \frac{\beta }{1-\beta }%
\frac{1}{M}$. However, the result that the limiting eigenvalue distribution
does not depend on $\beta $ (as long as $\beta \leq 1$) is connected with
the fact that
\begin{equation*}
\frac{1}{N}\mathbb{E}_{\beta ,N^{2}}\,\left( \mathrm{tr\,}\left( \frac{X_{N}%
}{N^{1/2}}\right) ^{2}\right) =1
\end{equation*}
for Curie-Weiss ensembles independent of $\beta \in
\mathbb{R}
$. In fact, whenever we have $\mathbb{E}(X_{N}(i,j)\,)=0$ and $%
\mathbb{E}(X_{N}(i,j)^{2}\,)=1$ the symmetry of the matrix implies%
\begin{equation*}
\frac{1}{N}\mathbb{E}\,\left( \mathrm{tr\,}\left( \frac{X_{N}}{N^{1/2}}%
\right) ^{2}\right) =\frac{1}{N^{2}}\sum_{i,j}\mathbb{E}(X_{N}(i,j)X_{N}(j,i))=1~.
\end{equation*}

Thus, whenever the limiting measure $\sigma $ exists (and has enough finite
moments) it must have second moment $\int t^{2}d\sigma =1$.

In this section we investigate the matrix norm
\begin{equation*}
\left\Vert A_{N}\right\Vert
=\left\Vert \frac{X_{N}}{N^{1/2}}\right\Vert = \max_{1\le i\le N}{|\lambda_i(A_{N})|} = \max\Big(|\lambda_1(A_N)|,|\lambda_N(A_N)|\Big)
\end{equation*}
for the Curie-Weiss and
related ensembles. For the \ `classical' Curie-Weiss ensemble $\mathbb{P}%
_{N^{2}}^{N^{2}\,F_{\beta }}$ we have:

\begin{proposition}
\bigskip \label{prop:Norm1}There is a constant $C$ such that for all $\beta
<1$
\begin{equation*}
\limsup_{N\rightarrow \infty }\,\mathbb{E}_{N^{2}}^{N^{2}\,F_{\beta
}}\left( \left\Vert A_{N}\right\Vert \right) \leq C
\end{equation*}
\end{proposition}

\begin{proof}
The expectation value of the matrix norm $\left\Vert A_{N}\right\Vert $ is
given by%
\begin{equation*}
\mathbb{E}_{N^{2}}^{N^{2}\,F_{\beta }}\left( \left\Vert A_{N}\right\Vert
\right) ~=~Z^{-1}\,\int_{-1}^{+1}\,E_{t}^{N^{2}}(\left\Vert A_{N}\right\Vert
)\,\frac{e^{-N^{2}\,F_{\beta }(t)\,/2}}{1-t^{2}}\,dt~\ \text{.}
\end{equation*}

Using the $N\times N$-matrix
\begin{equation*}
\mathcal{E}_{N}=\left(
\begin{array}{cccc}
1 & 1 & \ldots & 1 \\
1 & 1 & \ldots & 1 \\
\vdots & \vdots &  & \vdots \\
1 & 1 & \ldots & 1%
\end{array}%
\right) ~~
\end{equation*}
we estimate%
\begin{equation*}
E_{t}^{N^{2}}(\left\Vert A_{N}\right\Vert )~\leq ~E_{t}^{N^{2}}(\left\Vert
A_{N}-\frac{t}{\sqrt{N}}\,\mathcal{E}_{N}\right\Vert )\,+\,\frac{\left\vert
t\right\vert }{\sqrt{N}}\,\left\Vert \mathcal{E}_{N}\right\Vert ~.
\end{equation*}

The matrix $D_{N}=A_{N}-\frac{t}{\sqrt{N}}\,\mathcal{E}_{N}$ has random
entries $D_{N}(i,j)$ which are independent and have mean zero with respect
to the probability measure $P_{t}^{N^{2}}$. Thus we may apply \cite{Latala}
(after splitting $D_{N}$ into a lower and uper triangular part) and conclude
that $E_{t}^{N^{2}}(\left\Vert D_{N}\right\Vert )\leq C$ for a constant $%
C<\infty $.

The matrix $\mathcal{G}_{N}=\frac{1}{N}\,\mathcal{E}_{N}$ represents the
orthogonal projection onto the one dimensional subspace generated by the
vector $\eta _{N}=\frac{1}{\sqrt{N}}(1,1,\ldots ,1)$. Thus $\left\Vert
\mathcal{G}_{N}\right\Vert =1$ and $\left\Vert \mathcal{E}_{N}\right\Vert =N$%
. From Remark \ref{rem:ExpTbet} we learn that
\begin{equation*}
~Z^{-1}\,\int_{-1}^{+1}\left\vert t\right\vert \,\frac{e^{-N^{2}\,F_{\beta
}(t)\,/2}}{1-t^{2}}\,dt~\approx ~C_{1}\,\frac{1}{N}~.
\end{equation*}

Thus
\begin{equation*}
{\limsup }_{N\rightarrow \infty }\,\mathbb{E}_{N^{2}}^{N^{2}\,F_{\beta
}}\left( \left\Vert A_{N}\right\Vert \right) \leq {%
\limsup_{N\rightarrow \infty }}\;\big(C\,+\,C_{1}\frac{1}{\sqrt{N}}\big)\ =\ C~.
\end{equation*}
\end{proof}\\

The borderline case of generalized Curie-Weiss ensembles for Theorem~\ref%
{th:Hres} is the measure $\mathbb{E}_{N^{2}}^{N\,F_{\beta }}$. For this case
the expected value of the matrix norm does depend on $\beta $ and goes to
infinity as $\beta <1$ tends to $1$.

\begin{proposition}
\label{th:norm} For $\beta <1$ we have for positive constants $C_{1},\,C_{2}$%
\begin{eqnarray*}
\left( \frac{\beta }{1-\beta }\right) ^{\frac{1}{2}}\,C_{1}-C_{2} & \leq &
\liminf_{N\rightarrow \infty }\,\mathbb{E}_{N^{2}}^{N\,F_{\beta
}}\left( \left\Vert A_{N}\right\Vert \right) \\
 \leq  \limsup_{N\rightarrow \infty }\,\mathbb{E}%
_{N^{2}}^{N\,F_{\beta }}\left( \left\Vert A_{N}\right\Vert \right) &\leq &~
\left( \frac{\beta }{1-\beta }\right) ^{\frac{1}{2}}\,C_{1}+C_{2}~.
\end{eqnarray*}
\end{proposition}

\begin{proof}
The argument is close to the proof of the previous Proposition \ref%
{prop:Norm1}. We prove the lower bound, the upper bound is similar.
\\[2mm]
\noindent With the notation of the previous proof we have
\begin{eqnarray*}
E_{t}^{N^{2}}(\left\Vert A_{N}\right\Vert )~ &\geq &~\frac{\left\vert
t\right\vert }{\sqrt{N}}\,\left\Vert \mathcal{E}_{N}\right\Vert
-\,E_{t}^{N^{2}}(\left\Vert A_{N}-\frac{t}{\sqrt{N}}\,\mathcal{E}%
_{N}\right\Vert )\, \\
&\geq &~\left\vert t\right\vert \,\sqrt{N}-\,C_{2}
\end{eqnarray*}
using again the result of \cite{Latala} and $\left\Vert \mathcal{E}%
_{N}\right\Vert =N$.
\\[2mm]
Thus
\begin{eqnarray*}
\mathbb{E}_{N^{2}}^{N\,F_{\beta }}\left( \left\Vert A_{N}\right\Vert \right)
~ &\geq &~Z^{-1}\,\sqrt{N}\,\int_{-1}^{+1}\left\vert t\right\vert \,%
\frac{e^{-N\,F_{\beta }(t)\,/2}}{1-t^{2}}\,dt\  -\; C_{2}~.
\end{eqnarray*}

\noindent From Remark \ref{rem:ExpTbet} we learn that
\begin{equation*}
~Z^{-1}\,\int_{-1}^{+1}\left\vert t\right\vert \,\frac{e^{-N\,F_{\beta
}(t)\,/2}}{1-t^{2}}\,dt~\approx ~C_{1}\,\left( \frac{\beta }{1-\beta }%
\right) ^{1/2}\frac{1}{\sqrt{N}}
\end{equation*}
hence%
\begin{equation}
\liminf_{N\rightarrow \infty }\,\mathbb{E}_{N^{2}}^{N\,F_{\beta
}}\left( \left\Vert A_{N}\right\Vert \right)~\geq~\left( \frac{\beta }{1-\beta }\right) ^{\frac{1}{2}}\,C_{1}-C_{2}~.
\end{equation}
\end{proof}\\

We turn to the case of strong correlations, in particular, we consider the
full Curie-Weiss ensemble with inverse temperature $\beta >1$. It is easy to
see that for a full $\mathbb{\ }$Curie--Weiss ensemble $X_{N}(i,j)$ with
inverse temperature $\beta >1$ the `averaged traces'%
\begin{equation*}
\frac{1}{N}\;\mathbb{E}_{N^{2}}^{N^{2}\,F_{\beta }}\left( \mathrm{tr\,}%
\left( \frac{X_{N}}{N^{1/2}}\right) ^{k}\right)
\end{equation*}
cannot converge for $k$ large enough, in fact we have:

\begin{proposition}
\label{prop:traceTiefT}Consider the random matrix $B^{(\alpha )}=\frac{X_{N}%
}{N^{\alpha }}$, with $X_{N}$ symmetric and distributed according to the the
full Curie-Weiss ensemble with $\beta >1$. Then for $\alpha <1$ and $k\ $%
large enough and even we have%
\begin{equation*}
\frac{1}{N}\; \mathbb{E}_{N^{2}}^{N^{2}\,F_{\beta }}\,\left( \mathrm{tr\,}%
\left( \frac{X_{N}}{N^{\alpha }}\right) ^{k}\right) \rightarrow \infty
\qquad \text{as }N\rightarrow \infty
\end{equation*}
and for all $k\geq 1$%
\begin{equation*}
\frac{1}{N}\;\mathbb{E}_{N^{2}}^{N^{2}\,F_{\beta }}\,\left( \mathrm{tr\,}%
\left( \frac{X_{N}}{N}\right) ^{k}\right) \rightarrow 0\qquad \text{as }%
N\rightarrow \infty~.
\end{equation*}
\end{proposition}

\begin{proof}
We compute using \ (\ref{eq:mvonbeta})%
\begin{eqnarray*}
&& \frac{1}{N}\;\mathbb{E}_{N^{2}}^{N^{2}\,F_{\beta }}\,\left( \mathrm{tr\,}%
\left( \frac{X_{N}}{N^{\alpha }}\right) ^{k}\right) \\
&&=\frac{1}{N^{1+k\alpha }}\sum_{i_{1},i_{2},\ldots i_{k}}\mathbb{E}
_{N^{2}}^{N^{2}\,F_{\beta }}\left(
X_{N}(i_{1},i_{2})\,X_{N}(i_{2},i_{3})\,\ldots \,X_{N}(i_{k},i_{1})\right) \\
&&\geq \frac{1}{N^{1+k\alpha }}\sum_{\rho \left( i_{1},i_{2},\ldots
i_{k}\right) =k}\mathbb{E}_{N^{2}}^{N^{2}\,F_{\beta }}\left(
X_{N}(i_{1},i_{2})\,X_{N}(i_{2},i_{3})\,\ldots \,X_{N}(i_{k},i_{1})\right) \\
&&\geq \frac{1}{N^{1+k\alpha }}\,C\,N^{k}\,m(\beta )^{k} \quad\rightarrow
\quad \infty \qquad \text{for }k\text{ large,}
\end{eqnarray*}
where again $m(\beta )$ denotes the unique positive solution of $\;\tanh(\beta t)\,=\,t$.
We used above, that for all correlations
\begin{equation*}
\mathbb{E}
_{N^{2}}^{N^{2}\,F_{\beta }}\left(
X_{N}(i_{1},i_{2})\,X_{N}(i_{2},i_{3})\,\ldots \,X_{N}(i_{k},i_{1})\right)\geq 0\,.
\end{equation*}
\\[1mm]
The second assertion of the Proposition follows from%
\begin{eqnarray*}
&& \frac{1}{N}\;\mathbb{E}_{N^{2}}^{N^{2}\,F_{\beta }}\,\left( \mathrm{tr\,}%
\left( \frac{X_{N}}{N}\right) ^{k}\right) \\
&&=\frac{1}{N^{1+k}}\sum_{i_{1},i_{2},\ldots i_{k}}\mathbb{E}%
_{N^{2}}^{N^{2}\,F_{\beta }}\left(
X_{N}(i_{1},i_{2})\,X_{N}(i_{2},i_{3})\,\ldots \,X_{N}(i_{k},i_{1})\right) \\
&&\leq \frac{1}{N^{1+k}}\;N^{k}\ \rightarrow\ 0.
\end{eqnarray*}
Above we used that there are at most $N^{k}$ summand in the above sum.
\end{proof}\\

From Proposition \ref{prop:traceTiefT} we conclude that the eigenvalue
distribution function of $\frac{X_{N}}{N}$ converges to the Dirac measure $%
\delta _{0}$, while for $\frac{X_{N}}{N^{\alpha }}$ $(\alpha <1)$ at least
the moments do not converge. For $\beta >1$ the dependence (`interaction')
between the $X_{N}(i,j)$ is so strong that a macroscoping portion of the
random variables is aligned, i.e.\ either most of the $X_{N}(i,j)$ are
equal to $+1$ or most of the $X_{N}(i,j)$ are are equal to $-1$ and there
are about $m(\beta )N^{2}$ more aligned spins than others. Moreover, for
large $\beta $, the matrix $\frac{X_{N}}{N}$ should be \ close to the matrix
\begin{equation*}
\mathcal{G}_{N}=\frac{1}{N} \, \mathcal{E}_{N}\end{equation*}
or to $-\mathcal{G}_{N}$. This intuition is supported by the following
observation.

\begin{proposition}
Let $B_{N}=\frac{X_{N}}{N}$ with $X_{N}$ distributed according to $\mathbb{P}%
_{N^{2}}^{N^{2}\,F_{\beta }},$ then

\begin{enumerate}
\item \label{Teil1}For $\beta <1$ we have $\left\Vert B_{N}\right\Vert
\rightarrow 0$ in probability.

\item \label{Teil2}For $\beta >1$ we have $\left\Vert B_{N}\right\Vert
\rightarrow m\left( \beta \right) $ in probability.
\end{enumerate}
\end{proposition}

\begin{proof}
Part \ref{Teil1} follows from $\left\Vert B_{N}\right\Vert =\frac{1}{\sqrt{N}%
}\left\Vert A_{N}\right\Vert $ and from the estimate $\sup_{N}\mathbb{E}_{N^{2}}^{N^{2}\,F_{%
\beta }}\left( \left\Vert A_{N}\right\Vert \right) <\infty $ by Proposition %
\ref{th:norm}.\\[1mm]

To prove \ref{Teil2} we start with an estimate from below. We set $\eta _{N}=%
\frac{1}{\sqrt{N}}(1,\ldots ,1)$ and use the short hand notation $\mathbb{E}$ instead of
$\mathbb{E}_{N^{2}}^{N^{2}\,F_{\beta }}$.%
\begin{eqnarray*}
\mathbb{E}(\left\Vert B_{N}\right\Vert ^{2})~ &\geq &~\mathbb{E(}\left\Vert
B_{N}\,\eta _{N}\right\Vert ^{2}) \\
&=&\frac{1}{N^{3}}\,\sum_{i=1}^{N}\,\mathbb{E}\left\vert
\sum_{j=1}^{N}X_{N}(i,j)\right\vert ^{2} \\
&=&\frac{1}{N^{2}}\sum_{j,k=1}^{N}\,\mathbb{E}\Big(
X_{N}(1,j)\,X_{N}(1,k)\Big) \\
&=&\frac{1}{N^{2}}\,\Big( 1+N(N-1)\,\mathbb{E(}X_{N}(1,1)\,X_{N}(1,2)%
\Big) \,\rightarrow \,m(\beta )^{2}
\end{eqnarray*}
since $\mathbb{E}\Big(X_{N}(1,1)\,X_{N}(1,2)\Big)\rightarrow m(\beta )^{2}$ for $%
\beta >1$ by Proposition \ref{cor:CWCorr}. It follows that%
\begin{equation*}
\liminf_{N\to \infty}\mathbb{E}(\left\Vert B_{N}\,\eta _{N}\right\Vert ^{2k})~ \geq ~\liminf_{N\to \infty}\mathbb{E}%
(\left\Vert B_{N}\right\Vert ^{2})^{k} \geq ~m(\beta )^{2k}~.
\end{equation*}

\noindent We prove the converse inequality. For $k>1$ we have%
\begin{eqnarray*}
&&\mathbb{E}(\left\Vert B_{N}\right\Vert ^{2k})~ \leq ~\mathbb{E}(\mathrm{tr\,%
}B_{N}^{\,\,2k})  \notag \\
&=&\frac{1}{N^{2k}}\;\mathbb{E}\hspace*{0.01in}\left(
\sum_{i_{1},i_{2},\ldots
,i_{2k}}\,X_{N}(i_{1},i_{2})\,X_{N}(i_{2},i_{3})\,\ldots
\,X_{N}(i_{2k},i_{1})\right)~.
\end{eqnarray*}

Let $\rho =\rho (i_{1},i_{2},\ldots ,i_{2k})$ denote the number of different
indices among the $i_{j}$, i. e. $\rho (i_{1},i_{2},\ldots ,i_{2k})=\#\left\{
i_{1},i_{2},\ldots ,i_{2k}\right\} ,$ then%
\begin{equation*}
\frac{1}{N^{2k}}\;\mathbb{E}\left( \sum_{\rho \left(
i_{1},i_{2},\ldots ,i_{2k}\right)
<2k}\mkern-20mu  X_{N}(i_{1},i_{2}) \,  X_{N}(i_{2},i_{3})\,\ldots
\,X_{N}(i_{2k},i_{1})\right) \rightarrow 0
\end{equation*}
while%
\begin{equation*}
\frac{1}{N^{2k}}\;\mathbb{E}\left( \sum_{\rho \left(
i_{1},i_{2},\ldots ,i_{2k}\right)
=2k}\mkern-20mu X_{N}(i_{1},i_{2})\,X_{N}(i_{2},i_{3})\,\ldots
\,X_{N}(i_{2k},i_{1})\right) \rightarrow m(\beta )^{2k}
\end{equation*}
by Proposition \ref{cor:CWCorr}.

We also have
\begin{equation}
\mathbb{E}\left(\Vert B_N\Vert^{2}\right)~\leq~\mathbb{E}\left(\Vert B_N\Vert^{4}\right)^{1/2}
\end{equation}
 Thus we have proved that%
\begin{equation*}
\mathbb{E}(\left\Vert B_{N}\right\Vert ^{2k})~\rightarrow ~m(\beta )^{2k}
\end{equation*}
for all $k\in
\mathbb{N}
.$ It follows that $\left\Vert B_{N}\right\Vert ^{2}$ converges in
distribution to $\delta _{m(\beta )^{2}},$ hence $\left\Vert
B_{N}\right\Vert $ converges in distribution to $\delta _{m(\beta )},$
therefore it converges in probability to $m(\beta )$.
\end{proof}

\texttt{\bigskip }

\bigskip\bigskip

\noindent
\begin{tabular}{lcl}
\textbf{Winfried Hochst\"attler}&\quad& \texttt{winfried.hochstaettler@fernuni-hagen.de}\\
\textbf{Werner Kirsch}&\quad& \texttt{werner.kirsch@fernuni-hagen.de}\\
\textbf{Simone Warzel}&\quad& \texttt{warzel@ma.tum.de}
\end{tabular}
\end{document}